%% file: main.tex
\documentclass[letterpaper]{article}
\usepackage{times}
\usepackage{soul}
\usepackage{url}
\usepackage[hidelinks]{hyperref}
\usepackage[utf8]{inputenc}
\usepackage[small]{caption}
\usepackage{graphicx}
\usepackage{amsmath}
\usepackage{amsthm}
\usepackage{booktabs}
\usepackage{tikz}
\usepackage[linesnumbered,noend,ruled,noline]{algorithm2e}
\usepackage[numbers]{natbib}
\usepackage{geometry}
\usepackage{authblk}
\usepackage{amsfonts,amssymb,mathtools,color}
\usepackage{thmtools, thm-restate}

\newtheorem{theorem}{Theorem}

\date{}
\begin{document}
\title{PROPm Allocations of Indivisible Goods to Multiple Agents}
\author[a]{Artem Baklanov} \author[b]{Pranav Garimidi} \author[c]{Vasilis Gkatzelis} \author[c]{Daniel Schoepflin}
\affil[a]{Higher School of Economics, St. Petersburg}
\affil[b]{Columbia University}
\affil[c]{Drexel University}
\affil[ ]{{apbaklanov@hse.ru,  pg2682@columbia.edu, \{gkatz,schoep\}@drexel.edu}}
\date{}

\maketitle

\providecommand{\shortcite}[1]{\cite{#1}}

\newtheorem{corollary}{Corollary}[theorem]
\newtheorem{lemma}[theorem]{Lemma}
\newtheorem{observation}[theorem]{Observation}

\newcommand{\adjVal}[2][X]{\tilde{v}_{#2}(#1)}
\newcommand{\val}[3][X]{v_{#2}(#1_{#3})}
\newcommand{\minIn}[3][X]{m_{#2}(#1_{#3})}
\newcommand{\maxIn}[3][X]{M_{#2}(#1_{#3})}
\newcommand{\propXt}{\text{PROPm}}
\newcommand{\spx}{\text{CP}_i}
\newcommand{\spxi}[3]{\text{CP}_{#1}(#2,#3)}
\newcommand{\aEFx}{\text{a-EFx}}

\newcommand{\agents}{N}
\newcommand{\items}{M}
\newcommand{\allocationverbose}{(X_1, X_2, \dots, X_n)}
\newcommand{\minbundle}[2]{m_{#1}(#2)}
\newcommand{\minbundlei}[1]{\minbundle{i}{#1}}
\newcommand{\maxmin}[2]{d_{#1}(#2)}
\newcommand{\maxmini}{\maxmin{i}{X}}

\newcommand{\simplebundle}{{\text{bundle}}}
\newcommand{\bundleCollection}{\mathcal{A}}
\newcommand{\metabundle}{{\text{sub-problem}}}
\newcommand{\metaB}{(\bundleCollection, \agents')}
\newcommand{\propValue}{{\text{proportional share}}}
\newcommand{\metaallocation}{{\text{decomposition}}}
\newcommand{\incomplete}{{\text{incomplete}}}
\newcommand{\unassigned}{\text{unmatched}}
\newcommand{\lc}[2]{lc(#1, #2)}

\begin{abstract}
We study the classic problem of fairly allocating a set of indivisible goods among a group of agents, and focus on the notion of approximate proportionality known as PROPm. Prior work showed that there exists an allocation that satisfies this notion of fairness for instances involving up to five agents, but fell short of proving that this is true in general. We extend this result to show that a PROPm allocation is guaranteed to exist for all instances, independent of the number of agents or goods. Our proof is constructive, providing an algorithm that computes such an allocation and, unlike prior work, the running time of this algorithm is polynomial in both the number of agents and the number of goods.
\end{abstract}

\section{Introduction}
\input{introduction}

\section{Related Work}
\input{related_work}

\section{Our Results}
\input{ourResults}
\section{Preliminaries}
\input{preliminaries}

\section{PROPm Algorithm}
\input{NewAlgorithm}

\section{Correctness of the Algorithm}\label{sec:correctness}
\input{NewCorrectness}

\section{Running Time}\label{sec:running-time}
\input{polynomialTime}

\section{Conclusion}

\input{conclusion}

\section*{Acknowledgements}
The first author gratefully acknowledges support from the HSE University Basic Research Program. The last two authors were partially supported by NSF grants CCF-2008280 and CCF-2047907. We would also like to thank Maxim Timokhin for helpful feedback toward improving our algorithm.

\newpage
\bibliographystyle{named}
\bibliography{PROPm_Nagents}

\end{document}

%% file: introduction.tex
%We consider the well-studied problem of fairly distributing a set of scarce resources among a group of $n$ agents. This problem is at the heart of the long literature on fair division, initiated by \citet{Stern}, which has recently received renewed interest, partly due to the proliferation of automated resource allocation processes. To reach a fair outcome, such processes need to take into consideration the preferences of the agents, i.e., how much each agent values each of the resources. The most common modelling assumption regarding these preferences is that they are \emph{additive}: each agent $i$ has a value $v_{ij}\geq 0$ for each resource $j$, and her value for a set $S$ of resources is $v_i(S)=\sum_{j\in S}v_{ij}$. But, what would constitute a ``fair'' outcome given such preferences?

The fair allocation of scarce resources to a group of competing agents is a fundamental problem in both computer science and economics.  A particularly natural and well-studied setting is the fair allocation of indivisible goods to agents with additive valuations.  Under additive valuations, an agent $i$ has a value $v_{ij}$ for each good $j$ and her value for a bundle of goods $S$ is equal to the sum of the values over each good $j \in S$, i.e., $v_i(S) = \sum_{j \in S}{v_{ij}}$.  An indivisible good cannot be split and shared by more than one agent so achieving ``fairness'' with indivisible goods is often a difficult task.  Even determining the appropriate definition of fairness can be non-trivial.

One standard notion of fairness is \emph{proportionality}.  An allocation of a set of goods $M$ to $n$ agents is proportional if each agent $i$ receives a set of goods $S_i$ for which she has value $v_i(S_i) \geq \frac{1}{n}v_i(M)$.  In words, proportionality requires that every agent obtains at least a $1/n$ fraction of her total value.
%as much as she would if the total value were able to be split exactly equally (and, hence, exactly ``fairly'') among all agents.  
Unfortunately, when items are indivisible achieving proportionality may not be possible. For instance, when allocating a single indivisible good there is no way to provide \emph{any} positive value to anyone other than the one agent that receives the good. In fact, this example shows that one cannot even guarantee any \emph{multiplicative} approximation of proportionality. On the other hand, this instance does not rule out the existence of allocations satisfying \textit{additive} relaxations of proportionality.  

Three notable additive relaxations of proportionality are PROP1, PROPx, and PROPm. Each of these notions requires that agent $i$ must receive value no less than $\frac{1}{n}v_i(M) - d_i$ for some appropriately defined $d_i\geq 0$.  The least demanding of these notions is PROP1, wherein $d_i$ is the \emph{largest} value that agent $i$ has for any item allocated to another agent~\cite{CFS17}. On the other extreme, for PROPx $d_i$ is the \emph{smallest} value that agent $i$ has for any item allocated to another agent~\cite{Moulin2019}. PROP1 is known to be easy to satisfy and provides weak guarantees, while PROPx is overly demanding and known to not always exist. In the case of PROPm, $d_i$ corresponds to the \textit{maximin} value that agent $i$ has among items allocated to other agents~\cite{PROPm}. PROPm sits somewhere between PROP1 and PROPx, and it is the focus of this work.  %good not allocated to them to their bundle. In other words, the bundle each agent recieves is no more than being 1 good away from proportional. What exactly this 1 good can be varies between the notions. PROP1 allows this to be any good that an agent didn't receive while PROPx says this good can only be the least valuable item an agent didn't receive. Previous work has shown that PROP1 is quite easy to achieve and that there are instances where PROPx is impossible to satisfy. Thus, PROPm was proposed as a middle ground between the 2 notions. 

Baklanov \textit{et al.}~\shortcite{PROPm} demonstrated that there always exists a PROPm allocation for problem instances with up to five agents.  They also demonstrate that many other alternative relaxations of proportionality (e.g., letting $d_i$ be the value of the minimax value item, the median value item, and the average value item) fail to exist even for instances of three agents.  PROPm then seems to be a rather unique notion of approximate proportionality in that it strikes a balance between providing non-trivial guarantees and seemingly being plausible to exist in general cases.  However, the techniques used to prove this existence result required extensive case analysis, suggesting that they would not be useful toward an analogous proof for instances with many agents.  Two natural questions then arise from \cite{PROPm}: Are PROPm allocations always guaranteed to exist for any number of agents?  If so, can they be efficiently computed?  In this work, we answer both these questions in the affirmative.

%% file: related_work.tex
%The \emph{proportionality up to the most valued item} (PROP1) notion is a relaxation of proportionality that was introduced by \citet{CFS17}, who observed that there always exists a Pareto optimal allocation that satisfies PROP1. \citet{ACIW19} later extended  this notion to settings where the objects being allocated are chores, i.e., the valuations are negative, and very recently \citet{AMS2020} provided a strongly polynomial time algorithm for computing allocations that are Pareto optimal and PROP1 for both goods and chores. On the other extreme, it is known that the notion of \emph{proportionality up to the least valued item} (PROPx) may not be achievable even for small instances with three agents~\cite{Moulin2019,FSsurvey,AMS2020}.

As discussed above, it is impossible to guarantee any multiplicative approximation of proportionality in the indivisible items setting.  The first additive approximation, ``proportionality up to the most valued item'' (PROP1), was originally proposed by Conitzer \textit{et al.}~\shortcite{CFS17} where the authors demonstrated that there always exists a Pareto optimal allocation that is also PROP1.  On the other hand, Moulin~\shortcite{Moulin2019} showed that if we instead consider ``proportionality up to the least valued item'' (PROPx) we can no longer guarantee existence.  Moreover, Aziz \textit{et al.}~\shortcite{AMS2020} demonstrated that PROPx allocations may not exist even for instances with only three agents.

Another standard notion of fairness is ``envy-freeness'' wherein an agent is said to be envy-free if she has weakly higher value for the set of goods she receives than the set of goods any other agent receives. The instance with the single indivisible items, discussed above, verifies that envy-freeness may not be achievable either, so prior work has focused on notions of approximate envy-freeness, namely ``envy-freeness up to the most valued item'' (EF1)
\cite{budish2011combinatorial}
%\cite{Budish10} 
and ``envy-freeness up to the least valued item'' (EFx) \cite{CKMPS19}.  Similar to PROP1, EF1 allocations are known to exist for any number of agents \cite{Lipton}.  On the other hand, the existence or non-existence of EFx allocations has not been proven in general, and it is one of the main open problems in fair division. 

Plaut and Roughgarden~\shortcite{Plautt2018} demonstrated that EFx allocations always exist for two agents (even with combinatorial valuations) and Chaudhury \textit{et al.}~\shortcite{CGM2020} established the existence of EFx allocations for instances with three agents with additive valuations.  Extending the results in \cite{CGM2020} to more than three agents remains a challenging problem as the proof relies on complex case analysis, much like the proof of existence of PROPm allocations for up to five agents in \cite{PROPm}.  Central to many of the proofs of existence for EF1 and EFx is a variation of a procedure of Lipton \textit{et al.}~\shortcite{Lipton} known as ``envy-cycle elimination'' (see, e.g., \cite{Plautt2018,Chaudhury,oh2019fairly,Amanatidis_Markakis_Ntokos_2020}) wherein a graph representing a given allocation is constructed and an alternative allocation is produced by propagating changes along the edges of the graph.  Our algorithm for generating PROPm allocations has a very similar flavor, beginning from a partial allocation and using a graph analysis to imply a set of changes sufficient to arrive at a PROPm allocation.

Even if some fairness notion is shown to be achievable, it is still crucial to study the computational tractability of finding a solution that satisfies it.
%While the existence of some fairness notions is quite central in the literature, computability of fair allocations is of tantamount importance.
Aziz \textit{et al.}~\shortcite{AMS2020} provided a strongly polynomial-time algorithm producing a PROP1 and Pareto efficient allocation even in the presence of chores (i.e., some goods can have negative value).  For EF1 allocations, Caragiannis \textit{et al.}~\shortcite{CKMPS19} showed that maximizing the Nash social welfare (the geometric mean of the values of the agents) produces an allocation that is EF1 and Pareto efficient. On the other hand, Lee~\shortcite{Lee17} demonstrated that computing this is intractable.  However, the work of Barman \textit{et al.}~\shortcite{BarmanKV18} provided an alternative pseudo-polynomial time algorithm that computes an EF1 and Pareto optimal allocation.  For EFx, the picture is much less clear.  The algorithmic result in \cite{Plautt2018} relies on computing the allocation optimizing the leximin objective which may take exponential time and the result for three agents in \cite{CGM2020} leads only to a pseudo-polynomial time algorithm.  For PROPm, the existing results in \cite{PROPm} are constructive but may require exponential time in the number of items, even for just five agents.  On the other hand, in this work we demonstrate that PROPm allocations for any number of agents can, indeed, be computed in time polynomial in the number of agents and items -- a major improvement over \cite{PROPm}.

%\dnote{I'm not sure if we want to include anything for MMS, leaving it out for now}
%\anote{I agree with you}
%\emph{maximin share} (MMS) defined by \cite{Budish10} is another notion of fairness that is realated to PROPm. MMS asks that if an agent recieve value at least as much as if that agent was to partition all the items into groups and was given the set of items left over after every other agent picked a set. Thus when the agent is partitioning the items into groups, she looks to maximize the minimum share she would get from the worst element of her partition. Similar to many other notions of fairness for indivisible items, an allocation satisfying MMS does not always exist even for instances with three agents \cite{KPW2018}. Thus, just like for proportionality and envy freeness, attempts were made at satisfying relaxations to guarantee agents recieve some fraction of their maximin share. A polynomial time algorithm showing that agents could recieve $2/3$ of this benchmark was shown by \cite{AMNS2017}. This was followed up by \cite{BarmanM17} and \cite{GMT2018} with simpler algorithms for the smae result. These results continued to improved with a $3/4$ approximation in non-polynomial time and a $3/4-\epsilon$ poly-time approximation scheme by \cite{GhodsiHSSY18}. Finally the most recent work that has pushed this constant the fartherst is by \cite{Garg2020} who gave a polynomial time algorithm to find a $3/4$ approximation while also showing the existence of allocations with $3/4+1/12n$ approximations. 

%% file: ourResults.tex
%As demonstrated in \cite{PROPm}, PROPm is an interesting fairness notion that is significantly more demanding than PROP1, yet an allocation satisfying it exists for all instances involving five or fewer agents (unlike PROPx).  In fact, they consider a natural class of relaxation of proportionality and demonstrate that PROPm is the only one that does not fail to exist for instances with more than three agents.  
Prior to this work, we knew that an allocation satisfying PROPm always exists for instances involving up to five agents.
In this paper, we significantly extend this result by providing an algorithm that computes a PROPm allocation for any number of goods and agents. Moreover, our algorithm operates in time polynomial in both the number of agents and items, unlike the algorithm proposed in \cite{PROPm}, which was not polynomial even for a fixed number of agents. In light of these results, PROPm stands out as a rare example of a quite non-trivial fairness notion for which we get universal existence and polynomial-time computability.

Our algorithm employs a useful observation from \cite{PROPm} (see Observation \ref{obs:disjointAlloc} in Subsection 6.2 in this paper) which characterizes the conditions under which an instance can be split into agent- and item-disjoint sub-problems which can, effectively, be solved completely separately, yielding a full solution for the initial instance.  To produce such sub-problems, we consider a novel graph representation of our instance and search for paths through the graph. These paths imply a series of gradual modifications leading to the final decomposition of each problem instance into sub-problems.  We consider this algorithm to be of both practical and theoretical interest.  %PROPm allocations provide non-trivial fairness guarantees and we are able to find such allocations efficiently for any number of agents and items.  Moreover, 

%% file: preliminaries.tex
We study the problem of allocating a set $\items$ of $m$ indivisible items (or goods) to a set of $n$ agents $N=\{1,2,\dots,n\}$. Each agent $i$ has a value $v_{ij}\geq 0$ for each good $j$ and her value for receiving some subset of goods $S \subseteq \items$ is additive, i.e., $v_i(S) = \sum_{j \in S}{v_{ij}}$. For ease of presentation, we normalize the valuations so that $v_{i}(M) = 1$  for all $i\in \agents$. We also assume that $v_{ij} \leq 1/n$ for all $i \in N, j \in M$, because any item $j$ with $v_{ij}> 1/n$ could be assigned to $i$ and reduce the problem to finding a PROPm allocation of $M \setminus \{j\}$ to $N \setminus \{i\}$.\footnote{This fact is proven as Lemma 2 in \cite{PROPm}.} We let $\minbundlei{S} =\min_{j\in S}\{v_{ij}\}$ denote the value of the least valuable good for agent $i$ in bundle of goods $S$. 

An allocation $X = \allocationverbose$ is a partition of the goods into bundles % such that $X_1 \cup X_2 \cup \cdots \cup X_n = M$ and $X_k \cap X_\ell = \emptyset ~ \forall k,\ell$.  
such that $X_i$ is the bundle allocated to agent $i$.
%
%We now define certain notions of fairness on an allocation $X$. An allocation $X$ is Envy Free (EF) if $\val{i}{i}\ge\val{i}{j}$ $\forall i,j\in N$. Agent $i$ is said to EF envy agent $j$ if $\val{i}{j}>\val{i}{i}$ where $i$'s value for $j$'s bundle is higher than $i$'s value for their own bundle. An allocation $X$ is Envy Free up to any good (EFx) if $\val{i}{i}\ge \val{i}{j}-\min_{g \in X_{j}}(v_{i}(g))$ $\forall i,j \in N$. Agent $i$ is said to EFx envy agent $j$ if $\val{i}{i}\ge \val{i}{j}-\min_{g \in X_{j}}(v_{i}(g))$ where $i$'s value for $j$'s bundle is higher than $i$'s value for their own bundle even when $i$ removes their least favorite good from $g$'s bundle. 
%
%The lowest value good in each bundle from the perspective of each agent is important to track for our fairness objective and many other standard ones (e.g., EFx).  
We use $\maxmini = \max_{i' \neq i}\{\minbundlei{X_{i'}}\}$ to denote the value of the \emph{maximin good of agent $i$ in $X$}, and we say that an agent $i$ is \textbf{$\propXt$-satisfied} by $X$ if $v_i(X_i) + \maxmini \geq 1/n$. An allocation $X$ is $\propXt$ if it $\propXt$-satisfies every agent. 
The goal of our algorithm is to use these \simplebundle s to decompose the problem into smaller sub-problems, and compute a {\propXt} allocation using a divide \& conquer approach. 
A \textbf{\metabundle} $\metaB$ is a pair consisting of a set of  bundles $\bundleCollection = \{A_1, A_{2}, \dots, A_{k}\}$ and a subset of agents $\agents' \subseteq \agents$. 
%such that $|\bundleCollection| \geq |\agents'|$. 
In other words, a {\metabundle} ``matches'' a group of agents with a group of bundles, and our goal is going to be to do so in a way that computing a {\propXt} allocation for each {\metabundle} yields a {\propXt} allocation for the original problem. 
%\vnote{I would possibly defer the definitions of matched and unmatched sub-problems until later (or drop them altogether if we don't really need them)}
%We call a {\metabundle} \textbf{\unassigned} if $|\agents'| = 0$. 
The value of an agent $i$ for a set of bundles $\bundleCollection$ is $v_i(\bundleCollection)=\sum_{A_j \in \bundleCollection}{v_i(A_j)}$. We call a {\metabundle} $\metaB$ \textbf{proportional} if $v_i(\bundleCollection)/|\agents'|\geq 1/n$ for all $i\in \agents'$.

Given a set of  bundles $\bundleCollection$ and a set of agents $\agents'$, 
%with $|\agents'| \leq |\bundleCollection|$, 
a \textbf{\metaallocation} is a division of these agents and  bundles into (bundle and agent) disjoint  {\metabundle}s. %(where {\unassigned} \metabundle s are allowed). %In order to construct a recursive solution will seek ``high value'' \metaallocation s \pnote{... seek \metaallocation s where agents have "high value" for their corresponding meta-bundles}.
For example consider a set of five agents $N'=\{1,2,3,4,5\}$ and a set of five  bundles $\bundleCollection=\{A_1, A_2, A_3, A_4 ,A_5\}$.  One possible decomposition of $(A,N')$ would be into the  disjoint sub-problems
$(\{A_1, A_2, A_3\},\{1,2,3\})$ and $(\{A_4, A_5\},\{4,5\})$. 
%\vnote{I think that we need a clearer explanation/definition of sub-problems and decomposition. Let's provide a small concrete example.} 
We say that a {\metaallocation} for $\metaB$ is \textbf{proportional} if all of its included {\metabundle}s are proportional. 
As we show later on, as long as a {\metaallocation} is proportional, we can focus on solving each of its {\metabundle s} recursively without worrying about the allocation beyond that {\metabundle}. %\vnote{We could maybe also define a decomposition as a collection of disjoint sub-problems.}

Consider, again, the example above of five agents $\{1,2,3,4,5\}$ and five  bundles $\{A_1,...,A_5\}$. For these five agents assume agents 1, 2, and 3 have the same valuation function and assume agents 4 and 5 have the same valuation functions. Let the valuation functions for agents 1, 2, and 3 be $v(A_1)=\frac{1}{4}$ and $v(A_2)=v(A_3)=v(A_4)=v(A_5)=\frac{3}{16}$ and let the valuation function for agents 4 and 5 be $v(A_1)=v(A_2)=v(A_3)=\frac{1}{6}$, $v(A_4)=\frac{3}{20}$, and $v(A_5)=\frac{7}{20}$.  Since valuation functions are additive, we then have that $v_i(A_1\cup A_2\cup A_3)/3\ge 1/5$ for $i \in \{1,2,3\}$ and $v_i(A_4\cup A_5)/2\ge 1/5$ for $i \in \{4,5\}$. Thus, the  decomposition of $(A,N')$ described above $D = ((\{A_1, A_2, A_3\},\{1,2,3\}), (\{A_4, A_5\},\{4,5\}))$ is a proportional decomposition.  We will see that this means that we can solve the two sub-problems of allocating items in $A_1 \cup A_2 \cup A_3$ to agents 1, 2, and 3 such that they are $\propXt$-satisfied with respect to $A_1\cup A_2\cup A_3$ and allocating items in $A_4 \cup A_5$ to agents 4, 5 such that they are $\propXt$-satisfied with respect to $A_4\cup A_5$ to produce an allocation where every agent is $\propXt$-satisfied in the original problem. 

%\vnote{Should we define this as the ``sub-problem graph of a decomposition''? Do we use it for anything else?}
%Given a collection of disjoint sub-problems, we define the \textbf{sub-problem graph} to be a directed graph $G=(V, E)$, where each vertex in $V$ corresponds to a sub-problem and and an edge $(u,w)$ between two vertices $u,w\in V$ exists if and only if the corresponding sub-problems,  $(\bundleCollection_u, \agents_u)$ and $(\bundleCollection_w, \agents_w)$ satisfy the following condition: there exists some agent $k\in \agents_u$ who satisfies $\frac{v_k(\bundleCollection_w)}{|\agents_w|}\geq \frac{1}{n}$. 
%In other words, such an edge exists if and only if removing some agent from $\agents_w$ and replacing her with agent $k$ would maintain the proportionality of the $(\bundleCollection_w, \agents_w)$ sub-problem. 

%Note that this graph does not include all possible sub-problems, but rather just the sub-problems created by a run through of the algorithm. This means that the number of vertices and edges in the graph is bounded polynomially by the number of agents and items as is further detailed in Section \ref{sec:running-time}. Our algorithm uses this graph to gradually transform a given decomposition while maintaining its proportionality.

%% file: NewAlgorithm.tex
Our algorithm begins by choosing some arbitrary agent $i\in N$ to serve as the ``divider'' (we henceforth use $i$ to refer to the divider agent and $N^{-i}=N\setminus\{i\}$ to refer to the set of all other agents). The divider agent partitions the items into $n$ bundles, and then the algorithm proceeds to evaluate the other agents' preferences over these bundles to decide which one the divider should receive. Once the divider's bundle has been determined, the initial problem is decomposed into smaller sub-problems that are solved recursively.

\subsection*{Stage 1: the divider partitions the goods}
In order to partition the goods, the divider (agent $i$) first sorts them in non-decreasing order of value, from $i$'s perspective, and indexes them accordingly. Then, the first bundle $S_1$ corresponds to the longest prefix of goods in this ordering such that $v_i(S_1)\leq 1/n$. Observe that, by construction, $v_i(S_1) + v_{ij} > 1/n$ for all $j \in M\setminus S_1$.  Moreover, there is at least $(n-1)/n$ total value remaining for $i$ outside $S_1$.  We construct $S_{2}$ by taking the longest prefix of goods in $M \setminus S_1$ such that $i$ has value less than or equal to $1/(n-1) \cdot v_i(M \setminus S_1)$ for receiving all of them. Similarly, we let $S_k$ be the longest prefix of goods in $M \setminus (\cup_{j=1}^{k-1} S_j)$ such that the divider's value for these items remains less than or equal to $1/(n-k+1) \cdot v_i(M \setminus (\cup_{j=1}^{k-1} S_j))$.
%\anote{It should be $1/(n-k+1) \cdot v_i(....)$, right?}
%\vnote{Maybe add an example with actual numbers and, say, 2-3 agents? Not sure if this is worth it, though.}

\subsection*{Stage 2: decomposing into sub-problems}

Using disjoint bundles $S_1, S_2, \dots, S_n$ from the divider's partition, we now decompose the problem into sub-problems, eventually solving them recursively. Specifically, we carefully choose one of these $n$ bundles, say $S_t$, and allocate it to the divider. We then recursively allocate the items of bundles $S_1, \dots, S_{t-1}$ to some group $N_L$ of $t-1$ agents, and the items of bundles $S_{t+1},\dots, S_n$ to some group $N_R$ of $n-t$ agents.

%To choose which simple bundle to allocate to the divider agent, and how to decompose the agents in $N^{-i}$ into the two groups, $N_L$ and $N_R$, we gradually construct, and maintain, a proportional decomposition using a prefix of the simple bundles. This process is divided into (up to $n-1$) rounds, and at the end of each round we either have 

As the pseudocode of Algorithm~\ref{alg:PROPm} shows, the decomposition process works in a sequence of (up to) $n$ iterations, indexed by $t\in \{1, 2,\dots, n\}$. At the beginning of every iteration $t$, the algorithm has already identified a proportional decomposition $D$ involving $t-1$ agents and the  bundles $S_1, S_2, \dots, S_{t-1}$. At the end of step $t$, either the proportional decomposition $D$ has been updated to also include the bundle $S_t$ and a total of $t$ agents (possibly different than the $t-1$ ones that were participating in it at the beginning of the round), or the bundle $S_t$ has been assigned to the divider agent, and the remaining problem has been decomposed into a list of proportional sub-problems. Throughout the execution of the algorithm, $N_R$ is used to denote the set of agents that are not participating in the proportional decomposition $D$.

The first thing that the algorithm does in each iteration $t$ is to evaluate $c$, the number of agents from $N_R$ whose average value for the first $t$ bundles is more than $1/n$. If $c$ is equal to 0, this means that all the agents in $N_R$ essentially ``prefer'' sharing the last $n-t$ bundles instead of the first $t$ bundles. If this is the case, then the algorithm allocates bundle $S_t$ to the divider agent. It then recursively solves the proportional decomposition $D$, whose sub-problems involve $t-1$ agents and the first $t-1$ bundles, and also recursively solves the sub-problem involving the remaining $n-t$ agents (i.e., those in $N_R$) and the items from the last $n-t$ bundles, $S_{t+1}$ to $S_n$.

On the other hand, if the value of $c$ is positive, this suggests that there are agents in $N_R$ that ``prefer'' to share the first $t$ bundles rather than the last $n-t$ bundles. Intuitively, this suggests that the first $t$ bundles are ``over-demanded'', so our algorithm calls \textsc{UpdateDecomposition}, a crucial subroutine, to update decomposition $D$. 
%After each call to this subroutine, we update the entries for $D$, $N_R$, and $c$, which may have changed. 
As we discuss in subsection~\ref{sec:subroutine}, a single execution of this subroutine can have one of two possible outcomes: i) either the value of $c$ decreases by 1, or ii) the number of agents in the decomposition (denoted $|D.agents|$ for notational simplicity) increases by 1. The algorithm keeps calling this subroutine until either the decomposition grows to include bundle $S_t$ and $t$ agents, or $c$ drops to 0. In the former case, it continues to the next iteration (i.e., $t\leftarrow t+1$), otherwise, it assigns $S_t$ to the divider and recursively solves the remaining sub-problems.
%The algorithm keeps calling this subroutine until either $c$ becomes equal to 0 (at which point $S_t$ is assigned to the divider agent and the remaining sub-problems are solved recursively), or the new decomposition grows to include bundle $S_t$ and $t$ agents, at which point we proceed to the next iteration. 
Figure 1 shows an example of the algorithm running on a sample instance. 
%\vnote{Use some other notation for the number of agents in $D$ instead of $|D|$, which can be confusing given that $D$ is a set of sub-problems. It could be something like $|D.agents|$. Update that notation both in the algorithm as well as in the text.}

\begin{algorithm}%[H]
\DontPrintSemicolon
Let $S_1, S_2, \dots, S_n$ be the bundles the divider produces\;
Let $D$ be an, initially empty, decomposition\; 
%$S_0 \leftarrow \emptyset$\;
$N_R \leftarrow N^{-i}$\;
\For{$t=1$ \KwTo $n$}
{
    $c \leftarrow |\{k\in N_R: \frac{v_k(S_1 \cup \dots \cup S_t)}{t} > \frac{1}{n}\}|$ \;
     \While{$c>0$ and $|D.agents|<t$}
    {
        $D \leftarrow$ \textsc{UpdateDecomposition}\;
        $N_R \leftarrow$ subset of $N^{-i}$ not participating in $D$\;
        $c \leftarrow |\{k\in N_R: \frac{v_k(S_1 \cup \dots \cup S_t)}{t} > \frac{1}{n}\}|$ \;
    }
    \If{$|D.agents| < t$}
    {
        Allocate $S_t$ to the divider agent (agent $i$)\;
        Recursively solve all sub-problems of $D$\; \label{algline:leftsubproblems}
        Recursively solve $(S_{t+1} \cup \dots \cup S_n, N_R)$\; \label{algline:rightsubproblems}
        %\vnote{actually, $(S_{t+1}\cup\dots\cup S_n, N_R)$?}
        Return the combined allocation\;
    }
}
 \caption{PROPm Algorithm}\label{alg:PROPm}
\end{algorithm}

\subsection{The \textsc{UpdateDecomposition} subroutine}\label{sec:subroutine}

%OLD VERSION
%The \textsc{UpdateDecomposition} subroutine plays a central role in our PROPm algorithm, and it achieves the desired update of the existing proportional decomposition $D$ by leveraging its sub-problem graph structure. Note that whenever we call this subroutine, the value of $c$ is positive, so there exists at least one agent $k\in N_R$, i.e., not participating in $D$, for whom $v_k(S_1 \cup \dots \cup S_t)/t > 1/n$.

%Let $G$ be the sub-problem graph induced by $D$ in iteration $t$. The sub-problems corresponding to the vertices of this graph involve bundles $S_1,\dots, S_{t-1}$ and some set of $t-1$ agents (\vnote{so the number of vertices is at most $t-1$}). We add to this graph two more vertices. The first vertex, $w_\alpha$, corresponds to the agent $k\in N_R$ mentioned above; this vertex has outgoing edges to all the sub-problems $(\bundleCollection, \agents')$ of $D$ for which $\frac{v_k(\bundleCollection)}{|\agents'|}\geq \frac{1}{n}$. The second vertex, $w_\beta$, corresponds to the bundle $S_t$ that we wish to introduce to this decomposition. 

%NEW VERSION
The \textsc{UpdateDecomposition} subroutine plays a central role in our PROPm algorithm, and it achieves the desired update of the existing proportional decomposition $D$ at iteration $t$ by propagating changes on a carefully constructed graph.  Note that whenever we call this subroutine, the value of $c$ is positive, so there exists at least one agent $k \in N_R$, i.e., not participating in $D$, for whom $v_k(S_1 \cup \dots \cup S_t)/t > 1/n$.

    Given the decomposition of disjoint sub-problems $D$, we construct a directed ``sub-problem graph'' $G=(V, E)$, where each vertex in $V$ corresponds to a sub-problem in $D$ and and an edge $(u,w)$ between two vertices $u,w\in V$ exists if and only if the corresponding sub-problems,  $(\bundleCollection_u, \agents_u)$ and $(\bundleCollection_w, \agents_w)$ satisfy the following condition: there exists some agent $k\in \agents_u$ who satisfies $\frac{v_k(\bundleCollection_w)}{|\agents_w|}\geq \frac{1}{n}$. 
In other words, such an edge exists if and only if removing some agent from $\agents_w$ and replacing her with agent $k$ would maintain the proportionality of the $(\bundleCollection_w, \agents_w)$ sub-problem.

The sub-problems corresponding to the vertices of $G$ involve bundles $S_1,\dots, S_{t-1}$ and some set of $t-1$ agents so the number of vertices in $G$ is at most $t-1$. We add to $G$ two more vertices. The first vertex, $w_\alpha$, corresponds to the agent $k\in N_R$ mentioned above; this vertex has outgoing edges to all the sub-problems $(\bundleCollection, \agents')$ of $D$ for which $\frac{v_k(\bundleCollection)}{|\agents'|}\geq \frac{1}{n}$. The second vertex, $w_\beta$, corresponds to the bundle $S_t$ that we wish to introduce to this decomposition. 
%\pnote{This creates an \emph{unmatched} sub-problem where $|A|>|N'|$ with 1 bundle and 0 agents.} 
This vertex has incoming edges from any vertex whose sub-problem includes an agent $i'$ with value $v_{i'}(S_t)\geq 1/n$. In this graph, let $R$ be the set of vertices that are reachable from $w_\alpha$ via directed paths.

\textbf{Case 1.} If this set $R$ includes the vertex $w_\beta$, corresponding to the bundle $S_t$,  i.e., if there is a path from $w_\alpha$ to $w_\beta$, then the subroutine reallocates agents along the sub-problems of this path. Specifically, for each edge $(u,w)$ on this path, we remove from the sub-problem of $u$ the agent that is responsible for the existence of this edge (we choose one arbitrarily if there are multiple) and we place that agent in the sub-problem of $w$. As a result, $D$ would then include bundle $S_t$ as well as agent $k$, thus increasing $|D.agents|$ and, as we argue in Section~\ref{sec:correctness}, this modification maintains the proportionality of the decomposition.

\textbf{Case 2.} If the set $R$ does not include the vertex $w_\beta$, but it includes some agent $i'$ with $\frac{v_{i'}(S_1 \cup \cdots \cup S_t)}{t} \leq \frac{1}{n}$, then we perform an analogous shift of the agents across the sub-problems along the path from $k$ to $i'$, but remove agent $i'$ from the decomposition and add her to the set $N_R$. This, again, does not compromise the proportionality of the decomposition, but it ensures that the updated value of $c$ will drop by 1 since agent $k$ was removed from $N_R$ and replaced with agent $i'$ who does not contribute toward an increase of the value of $c$.

\textbf{Case 3.} Finally, if neither of the cases above holds, the subroutine takes all the agents and all the bundles corresponding to vertices in $R$ and merges them into a single sub-problem, together with agent $k$ and bundle $S_t$. This, again, increases $|D.agents|$ and as we show using a separate argument in Section~\ref{sec:correctness}, it maintains the proportionality of the decomposition.

\tikzset{every picture/.style={line width=0.75pt}} %set default line width to 0.75pt        

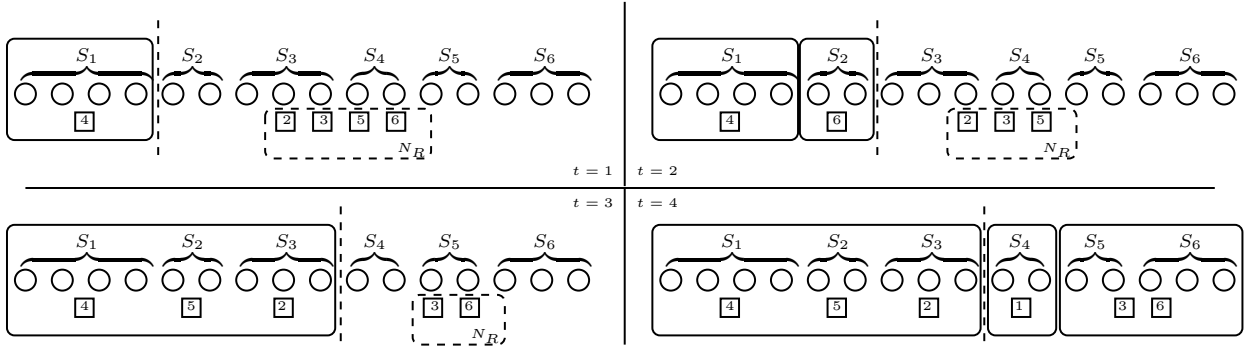
\begin{figure*}

\begin{tikzpicture}[x=0.75pt,y=0.75pt,yscale=-1,xscale=1, scale =.93]

\draw (20,50) circle (.15cm);
\draw (40,50) circle (.15cm);
\draw (60,50) circle (.15cm);
\draw (80,50) circle (.15cm);
\draw (100,50) circle (.15cm);
\draw (120,50) circle (.15cm);
\draw (140,50) circle (.15cm);
\draw (160,50) circle (.15cm);
\draw (180,50) circle (.15cm);
\draw (200,50) circle (.15cm); 
\draw (220,50) circle (.15cm); 
\draw (240,50) circle (.15cm); 
\draw (260,50) circle (.15cm); 
\draw (280,50) circle (.15cm); 
\draw (300,50) circle (.15cm); 
\draw (320,50) circle (.15cm); 
\draw [dashed] (92,10) -- (92,85);

\draw[rounded corners = 1mm] (10,20) rectangle (89,75);
\draw (47,60) rectangle (57,70);
\draw (48,60) node [anchor=north west][inner sep=0.75pt]  [font=\tiny]  {$4$};

\draw (156,60) rectangle (166,70);
\draw (176,60) rectangle (186,70); 
\draw (196,60) rectangle (206,70); 
\draw (216,60) rectangle (226,70); 

\draw[rounded corners = 1mm,dashed] (150,58) rectangle (240,85);
% Text Node
\draw (220,75) node [anchor=north west][inner sep=0.75pt]  [font=\tiny]  {$N_R$};

\draw (315,87.5) node [anchor=north west][inner sep=0.75pt]  [font=\tiny]  {$t=1$};

\draw (157.86,60) node [anchor=north west][inner sep=0.75pt]  [font=\tiny]  {$2$};
% Text Node
\draw (177.86,60) node [anchor=north west][inner sep=0.75pt]  [font=\tiny]  {$3$};
% Text Node
\draw (198,60) node [anchor=north west][inner sep=0.75pt]  [font=\tiny]  {$5$};
% Text Node
\draw (216.5,60) node [anchor=north west][inner sep=0.75pt]  [font=\tiny]  {$6$};

%\draw (10,5) node [anchor=north west][inner sep=0.75pt]  [font=\small]  {$t=1:$};

% Text Node
\draw (16,22) node [anchor=north west][inner sep=0.75pt]    {$\overbrace{\ \ \ \ \ \ \ \ \ \ \ \ \ \ \ \ \ \ \ \ }^{S_{1}}$};
% Text Node
\draw (93,22) node [anchor=north west][inner sep=0.75pt]    {$\overbrace{\ \ \ \ \ \ \ \ \ }^{S_{2}}$};
% Text Node
\draw (135,22) node [anchor=north west][inner sep=0.75pt]    {$\overbrace{\ \ \ \ \ \ \ \ \ \ \ \ \ \ }^{S_{3}}$};
% Text Node
\draw (195,22) node [anchor=north west][inner sep=0.75pt]    {$\overbrace{\ \ \ \ \ \ }^{S_{4}}$};
% Text Node
\draw (234,22) node [anchor=north west][inner sep=0.75pt]    {$\overbrace{\ \ \ \ \ \ \ \ }^{S_{5}}$};
% Text Node
\draw (275,22) node [anchor=north west][inner sep=0.75pt]    {$\overbrace{\ \ \ \ \ \ \ \ \ \ \ \ \ \ }^{S_{6}}$};


\draw (370,50) circle (.15cm);
\draw (390,50) circle (.15cm);
\draw (410,50) circle (.15cm);
\draw (430,50) circle (.15cm);
\draw (450,50) circle (.15cm);
\draw (470,50) circle (.15cm);
\draw (490,50) circle (.15cm);
\draw (510,50) circle (.15cm);
\draw (530,50) circle (.15cm);
\draw (550,50) circle (.15cm); 
\draw (570,50) circle (.15cm); 
\draw (590,50) circle (.15cm); 
\draw (610,50) circle (.15cm); 
\draw (630,50) circle (.15cm); 
\draw (650,50) circle (.15cm); 
\draw (670,50) circle (.15cm); 
\draw[dashed]  (482,10) -- (482,85);

\draw  (345,0) -- (345,100);

%\draw (360,5) node [anchor=north west][inner sep=0.75pt]  [font=\small]  {$t=2:$};

\draw[rounded corners = 1mm] (360,20) rectangle (439,75);
\draw (397,60) rectangle (407,70);
\draw (398,60) node [anchor=north west][inner sep=0.75pt]  [font=\tiny]  {$4$};

\draw[rounded corners = 1mm] (440,20) rectangle (480,75);
\draw (455,60) rectangle (465,70); 
\draw (456,60) node [anchor=north west][inner sep=0.75pt]  [font=\tiny]  {$6$};

\draw (526,60) rectangle (536,70); 
\draw (546,60) rectangle (556,70); 
\draw (566,60) rectangle (576,70); 

\draw[rounded corners = 1mm,dashed] (520,58) rectangle (590,85);
% Text Node
\draw (570,75) node [anchor=north west][inner sep=0.75pt]  [font=\tiny]  {$N_R$};

\draw (350,87.5) node [anchor=north west][inner sep=0.75pt]  [font=\tiny]  {$t=2$};

% Text Node
\draw (527,60) node [anchor=north west][inner sep=0.75pt]  [font=\tiny]  {$2$};
% Text Node
\draw (548,60) node [anchor=north west][inner sep=0.75pt]  [font=\tiny]  {$3$};
% Text Node
\draw (566.5,60) node [anchor=north west][inner sep=0.75pt]  [font=\tiny]  {$5$};

% Text Node
\draw (366,22) node [anchor=north west][inner sep=0.75pt]    {$\overbrace{\ \ \ \ \ \ \ \ \ \ \ \ \ \ \ \ \ \ \ \ }^{S_{1}}$};
% Text Node
\draw (443,22) node [anchor=north west][inner sep=0.75pt]    {$\overbrace{\ \ \ \ \ \ \ \ \ }^{S_{2}}$};
% Text Node
\draw (485,22) node [anchor=north west][inner sep=0.75pt]    {$\overbrace{\ \ \ \ \ \ \ \ \ \ \ \ \ \ }^{S_{3}}$};
% Text Node
\draw (545,22) node [anchor=north west][inner sep=0.75pt]    {$\overbrace{\ \ \ \ \ \ }^{S_{4}}$};
% Text Node
\draw (584,22) node [anchor=north west][inner sep=0.75pt]    {$\overbrace{\ \ \ \ \ \ \ \ }^{S_{5}}$};
% Text Node
\draw (625,22) node [anchor=north west][inner sep=0.75pt]    {$\overbrace{\ \ \ \ \ \ \ \ \ \ \ \ \ \ }^{S_{6}}$};

\end{tikzpicture}

\nointerlineskip

\begin{tikzpicture}[x=0.75pt,y=0.75pt,yscale=-1,xscale=1, scale = .93]

\draw (20,50) circle (.15cm);
\draw (40,50) circle (.15cm);
\draw (60,50) circle (.15cm);
\draw (80,50) circle (.15cm);
\draw (100,50) circle (.15cm);
\draw (120,50) circle (.15cm);
\draw (140,50) circle (.15cm);
\draw (160,50) circle (.15cm);
\draw (180,50) circle (.15cm);
\draw (200,50) circle (.15cm); 
\draw (220,50) circle (.15cm); 
\draw (240,50) circle (.15cm); 
\draw (260,50) circle (.15cm); 
\draw (280,50) circle (.15cm); 
\draw (300,50) circle (.15cm); 
\draw (320,50) circle (.15cm); 
\draw [dashed] (191,10) -- (191,85);

\draw  (345,0) -- (345,85);
\draw (20,0) -- (665,0); 

%\draw (10,5) node [anchor=north west][inner sep=0.75pt]  [font=\small]  {$t=3:$};

\draw[rounded corners = 1mm] (10,20) rectangle (188,80);
\draw (47,60) rectangle (57,70);
\draw (48,60) node [anchor=north west][inner sep=0.75pt]  [font=\tiny]  {$4$};

%\draw[rounded corners = 1mm] (90,20) rectangle (130,80);
\draw (105,60) rectangle (115,70); 
\draw (106,60) node [anchor=north west][inner sep=0.75pt]  [font=\tiny]  {$5$};

%\draw[rounded corners = 1mm] (132,20) rectangle (188,80);
\draw (155,60) rectangle (165,70); 
\draw (155,60) node [anchor=north west][inner sep=0.75pt]  [font=\tiny]  {$2$};

\draw (236,60) rectangle (246,70); 
\draw (256,60) rectangle (266,70); 

\draw[rounded corners = 1mm,dashed] (230,58) rectangle (280,85);
% Text Node
\draw (260,75) node [anchor=north west][inner sep=0.75pt]  [font=\tiny]  {$N_R$};

\draw (315,4) node [anchor=north west][inner sep=0.75pt]  [font=\tiny]  {$t=3$};

% Text Node
\draw (238,60) node [anchor=north west][inner sep=0.75pt]  [font=\tiny]  {$3$};
% Text Node
\draw (256.5,60) node [anchor=north west][inner sep=0.75pt]  [font=\tiny]  {$6$};

% Text Node
\draw (16,22) node [anchor=north west][inner sep=0.75pt]    {$\overbrace{\ \ \ \ \ \ \ \ \ \ \ \ \ \ \ \ \ \ \ \ }^{S_{1}}$};
% Text Node
\draw (93,22) node [anchor=north west][inner sep=0.75pt]    {$\overbrace{\ \ \ \ \ \ \ \ \ }^{S_{2}}$};
% Text Node
\draw (135,22) node [anchor=north west][inner sep=0.75pt]    {$\overbrace{\ \ \ \ \ \ \ \ \ \ \ \ \ \ }^{S_{3}}$};
% Text Node
\draw (195,22) node [anchor=north west][inner sep=0.75pt]    {$\overbrace{\ \ \ \ \ \ }^{S_{4}}$};
% Text Node
\draw (234,22) node [anchor=north west][inner sep=0.75pt]    {$\overbrace{\ \ \ \ \ \ \ \ }^{S_{5}}$};
% Text Node
\draw (275,22) node [anchor=north west][inner sep=0.75pt]    {$\overbrace{\ \ \ \ \ \ \ \ \ \ \ \ \ \ }^{S_{6}}$};


\draw (370,50) circle (.15cm);
\draw (390,50) circle (.15cm);
\draw (410,50) circle (.15cm);
\draw (430,50) circle (.15cm);
\draw (450,50) circle (.15cm);
\draw (470,50) circle (.15cm);
\draw (490,50) circle (.15cm);
\draw (510,50) circle (.15cm);
\draw (530,50) circle (.15cm);
\draw (550,50) circle (.15cm); 
\draw (570,50) circle (.15cm); 
\draw (590,50) circle (.15cm); 
\draw (610,50) circle (.15cm); 
\draw (630,50) circle (.15cm); 
\draw (650,50) circle (.15cm); 
\draw (670,50) circle (.15cm);

\draw (350,4) node [anchor=north west][inner sep=0.75pt]  [font=\tiny]  {$t=4$};

\draw[rounded corners = 1mm] (360,20) rectangle (538,80);
\draw (397,60) rectangle (407,70);
\draw (398,60) node [anchor=north west][inner sep=0.75pt]  [font=\tiny]  {$4$};

\draw (455,60) rectangle (465,70); 
\draw (456,60) node [anchor=north west][inner sep=0.75pt]  [font=\tiny]  {$5$};

\draw (505,60) rectangle (515,70); 
\draw (505,60) node [anchor=north west][inner sep=0.75pt]  [font=\tiny]  {$2$};

%OLD VERSION
%\draw [dashed] (581,10) -- (581,85);
%\draw[rounded corners = 1mm] (540,20) rectangle (579,80);
%\draw[rounded corners = 1mm] (583,20) rectangle (680,80);

%NEW VERSION
\draw [dashed] (540,10) -- (540,85);
\draw[rounded corners = 1mm] (542,20) rectangle (579,80);
\draw[rounded corners = 1mm] (581,20) rectangle (680,80);

\draw (611,60) rectangle (621,70); 
\draw (631,60) rectangle (641,70); 

% Text Node
\draw (611,60) node [anchor=north west][inner sep=0.75pt]  [font=\tiny]  {$3$};
% Text Node
\draw (631.5,60) node [anchor=north west][inner sep=0.75pt]  [font=\tiny]  {$6$};

\draw (555,60) rectangle (565,70); 
\draw (555,60) node [anchor=north west][inner sep=0.75pt]  [font=\tiny]  {$1$};

% Text Node
\draw (366,22) node [anchor=north west][inner sep=0.75pt]    {$\overbrace{\ \ \ \ \ \ \ \ \ \ \ \ \ \ \ \ \ \ \ \ }^{S_{1}}$};
% Text Node
\draw (443,22) node [anchor=north west][inner sep=0.75pt]    {$\overbrace{\ \ \ \ \ \ \ \ \ }^{S_{2}}$};
% Text Node
\draw (485,22) node [anchor=north west][inner sep=0.75pt]    {$\overbrace{\ \ \ \ \ \ \ \ \ \ \ \ \ \ }^{S_{3}}$};
% Text Node
\draw (545,22) node [anchor=north west][inner sep=0.75pt]    {$\overbrace{\ \ \ \ \ \ }^{S_{4}}$};
% Text Node
\draw (584,22) node [anchor=north west][inner sep=0.75pt]    {$\overbrace{\ \ \ \ \ \ \ \ }^{S_{5}}$};
% Text Node
\draw (625,22) node [anchor=north west][inner sep=0.75pt]    {$\overbrace{\ \ \ \ \ \ \ \ \ \ \ \ \ \ }^{S_{6}}$};
\end{tikzpicture}
\caption{The state of our algorithm after the completion of iterations $t\in \{1,2,3,4\}$ on a sample instance with six agents and agent 1 as the divider.  The circles correspond to items, which are grouped together into bundles $\{S_1,\dots,S_6\}$ by the divider. The numbered boxes beneath the items correspond to agents 2 through 6, and each of the larger rounded rectangles, containing bundles and agents, correspond to sub-problems. In each iteration of the algorithm, captured by this figure, the dashed vertical line separates the bundles of the decomposition $D$, on the left, from the remaining bundles. After each iteration $t$, either the decomposition is updated to include bundle $S_t$ (as we see for $t\in \{1,2,3\}$), or bundle $S_t$ is allocated to the divider agent (as wee see for $t=4$), finalizing the set of sub-problems to be solved recursively.}
    \label{fig:my_label}
\end{figure*}

%% file: NewCorrectness.tex
To verify the correctness of the algorithm, we first show that it always terminates and returns an allocation (in fact, we demonstrate in Section \ref{sec:running-time} that the algorithm completes in polynomial time). As we verify in subsection~\ref{sec:correctSubroutine}, a single call to the \textsc{UpdateDecomposition} subroutine returns an updated decomposition that has either one additional agent (and bundle) or has reduced the value of $c$ by 1. Since the value of $c$ at the beginning of each iteration can never be more than $n-1$, this ensures that the while loop will always terminate within a finite number of iterations. If for some iteration $t\leq n-1$ the value of $|D.agents|$ drops below $t$, then the algorithm recurses on smaller problems and returns the induced allocation. If, on the other hand, $|D.agents|$ does not drop below $t$ for any iteration $t\leq n-1$, then when $t=n$ we have an empty set $N_R$, necessarily leading to $D$ including every agent except for the divider (meaning $|D.agents|=n-1 < t$). Also, note that the size of the sub-problems solved recursively always strictly decreases.  We demonstrate in subsection \ref{sec:satisfied} that this process yields a proportional allocation for all agents.

\subsection{Correctness of \textsc{UpdateDecomposition}}\label{sec:correctSubroutine}
We now formally prove that after every execution of the \textsc{UpdateDecomposition} subroutine, either $|D.agents|$ increases by 1, or the value of $c$ drops by 1. In both cases, the resulting decomposition remains proportional, throughout the execution of the algorithm. 

First, note that the set $R$ of vertices that are reachable from agent $k$ is bound to be non-empty. This is due to the fact that $v_k(S_1\cup\dots\cup S_t)/t> 1/n$, i.e., the agent's average bundle value is more than $1/n$ and, by pigeonhole principle, there must exist some sub-problem $(\bundleCollection, \agents')$ such that $v_k(\bundleCollection)/|\bundleCollection| \geq 1/n$. Since the set $R$ is not empty, the description of the subroutine in the previous section clearly shows that it always achieves either an increase of $|D.agents|$ or a drop of the value of $c$. Therefore, the rest of this subsection focuses on proving that all of these updates on the decomposition maintain its proportionality.
%\anote{Should we use consistent notation: either $v_k(S_1\cup\dots\cup S_t)$ or $v_k(S_1, \dots, S_t)?$ The second one allows $v_k(\bundleCollection).$ }
\begin{lemma}
Given a proportional decomposition, the described re-allocation of agents across the directed edges of a path in the decomposition's sub-problem graph leads to a new decomposition that remains proportional.
\end{lemma}
\begin{proof}
Note that, by definition of the sub-problem graph, any agent that caused the existence of an edge $(u, w)$ must have a value of at least $|N_w|/n$ for the bundles of the sub-problem corresponding to vertex $w$. As a result that agent will still satisfy proportionality if moved from $u$ to $w$. This is true for all the edges on this graph (including the ones connecting the two added vertices $w_\alpha$ and $w_\beta$), ensuring the proportionality is maintained.
\end{proof}

The more demanding case is to verify that proportionality is also maintained by the third type of update that this subroutine performs, i.e., the creation of a sub-problem involving the agents and goods in $R$ as well as agent $k$ and bundle $S_t$.

\begin{lemma}\label{lem:value_bound}
If $\bundleCollection'$ is the collection of all the  bundles corresponding to sub-problems not reachable from $w_\alpha$, excluding $S_t$, then for every agent $q$ from a sub-problem in $R$ we have $v_{q}(\bundleCollection')/|\bundleCollection'| < 1/n$. The same is true for the agent $k$ that corresponds to vertex $w_\alpha$, i.e., $v_k(\bundleCollection')/|\bundleCollection'| < 1/n$.
\end{lemma}
%\begin{lemma}
%Suppose that the {\metabundle}s in $D_j \cup \hat{D}$ consist of $t$ simple bundles.  Then all agents $k$ from vertices in $D$ have $v_{k}(D_j \cup \hat{D}) < t/n$.
%\end{lemma}
\begin{proof}
Assume that this is not the case, i.e., that either some agent $q$ corresponding to a sub-problem in $R$, or agent $k$, has an average bundle value at least $1/n$ for the bundles in $\bundleCollection'$. By the pigeonhole principle, this implies that there must exist some sub-problem not reachable from $R$ such that this agent's average value for the bundles in that sub-problem is at least $1/n$. But, based on the definition of the sub-problem graph, this would imply the existence of an edge from that agent's vertex to this sub-problem's vertex, contradicting the fact that the latter is not reachable from the former.
%We proceed via contradiction.  Suppose that some agent $k'$ from a vertex in $D$ has $v_{k'}(D_j \cup \hat{D}) \geq t/n$.  We first note by our argument above that for all agents $k$ in {\metabundle}s corresponding to vertices in $D$ that $v_k(D_j) < 1/n$.  It must then be that $v_{k'}(\hat{D}) > (t-1)/n$.  Since this is a {\metaallocation} there are at most $t-1$ agents assigned to bundles in $\hat{D}$.  Moreover, no {\metabundle}s in $\hat{D}$ are {\unassigned} since this is a proportional {\incomplete} {\metaallocation} and $D_j$ is the {\unassigned} {\metabundle}.  However, we note that $\hat{D}$ has no incoming edges, so there must not be any agent who would retain {\propValue} at least $1/n$ if she were to replace an agent from one of the {\metabundle}s in $\hat{D}$.  In particular, this is true for agent $k'$.  So the value that $k'$ has for any {\metabundle} $\metaB$ in $\hat{D}$ is strictly less than $\frac{1}{n} \cdot |N'|$.  Since there are weakly fewer than $t-1$ agents assigned to {\metabundle}s in $\hat{D}$ we have that $v_{k'}(\hat{D}) \leq (t-1)/n$, a contradiction.
\end{proof}

%\begin{lemma}
%Suppose that the {\metabundle}s in $D_j \cup \hat{D}$ consist of $t$ simple bundles.  Then agent $i'$ has $v_{i'}(D_j \cup \hat{D}) < t/n$.
%\end{lemma}
%\begin{proof}
%The proof proceeds nearly identically as above.  It must be that $v_{i'}(D_j) < 1/n$ by definition.  Moreover, since $\hat{D} \cap D_1 = \emptyset$ we know that $i'$ would not have $\propValue$ at least $1/n$ if she were to replace some agent in any {\metabundle} corresponding to vertices in $\hat{D}$. Thus we know that the value $i'$ has for the items in any {\metabundle} $\metaB$ in $\hat{D}$ is strictly less than by $\frac{1}{n} \cdot |N'|$. Then, since $\hat{D}$ is comprised of $t-1$ simple bundles, it must be that $v_{i'}(\hat{D}) < (t-1)/n$.
%\end{proof}

Whenever the subroutine resorts to the third type of decomposition update (Case 3), this means that there is no agent $i'$ in a sub-problem in $R$ such that $\frac{v_{i'}(S_1 \cup \dots \cup S_t)}{t} \leq \frac{1}{n}$ (Case 2). Thus,  agent $k$ and all agents $q$ from sub-problems in $R$ have value for the items in $S_1\cup \dots \cup S_t$ greater than $t/n$. Lemma~\ref{lem:value_bound} implies that their total value for the  bundles in $\bundleCollection'$ is less than $|\bundleCollection'|/n$. Therefore, these agents' value for the items contained in sub-problems in $R$ is at least $(t-|\bundleCollection'|)/n$. 

Also, note that the overall number of agents in the sub-problems of  $G$ (excluding $k$) is $t-1$, and the number of agents in the sub-problems not reachable from $R$ are $|\bundleCollection'|$, so the total number of agents in the sub-problems of $R$ is $t-|\bundleCollection'|-1$. Therefore, if we define a new sub-problem using the agents from $R$, combined with agent $k$ corresponding to vertex $w_\alpha$, and the  bundles from $R$, combined with $S_t$, then this sub-problem will be proportional, because every agent's value for the  bundles in it will be at least $(t-|\bundleCollection'|)/n$ and the number of agents in it is $t-|\bundleCollection'|$. 

\subsection{All agents are PROPm-satisfied} \label{sec:satisfied}
To verify that the induced allocation is always PROPm, we first restate a useful observation from \cite{PROPm}. This observation provides us with a sufficient condition under which ``locally'' satisfying $\propXt$ in each sub-problem yields a ``globally'' $\propXt$ allocation. Given an allocation of a subset of items to a subset of agents, we say that this partial allocation is $\propXt$ if the agents involved would be $\propXt$-satisfied if no other agents or items were present.
%We are able to find allocations more readily because $\propXt$ is a more ``localized'' property than, say, EFx. 
%Provided that some subset of agents has enough total value for a given subset of items, we can recursively solve the problem on this smaller set of agents.  More concretely, at several points in the construction of $\propXt$ allocations we will use the following observation. 

\begin{observation}\label{obs:disjointAlloc}
Let $N_1, N_2$ be two disjoint sets of agents, let $M_1$ and $M_2 = M \setminus M_1$ be a partition of the items into two sets, and let $X$ be an allocation of the items in $M_1$ to agents in $N_1$ and items in $M_2$ to agents in $N_2$. Then, if some agent $i \in N_1$ is $\propXt$-satisfied with respect to the partial allocation of the items in $M_1$ to the agents in $N_1$, and $\frac{v_i(M_1)}{|N_1|} \geq \frac{1}{|N_1 + N_2|}$, then $i$ is $\propXt$-satisfied by $X$ regardless of how the items in $M_2$ are allocated to agents in $N_2$.  
\end{observation}

We now prove a result regarding the partition implied by the divider agent's preferences, which is analogous to a theorem that is shown by \cite{PROPm} for a different, much more complicated, partition of the items.
\begin{theorem}\label{thm:recSpx}
%Let $S_1, S_{2}, \dots, S_n$ be the  bundles defined by the divider agent's partition. 
If the divider agent receives any bundle $S_{\ell}$ and no item from $S_1 \cup S_{2} \cup \cdots \cup S_{\ell - 1}$ is allocated to the same agent as an item from $S_{\ell + 1} \cup S_{\ell + 2} \cup \cdots \cup S_n$, then agent $i$ will be $\propXt$-satisfied.
%The agent will be PROPm satisfied if she receives any bundle $S_{\ell}$ for any allocation of the remaining items as long as the items in $S_n \cup S_{n-1} \cup \cdots \cup S_{\ell + 1}$ are in disjoint bundles from the items in $S_{\ell - 1} \cup S_{\ell - 2} \cup \cdots \cup S_1$. 
\end{theorem}
\begin{proof}
%The proof of this statement is not much harder than Observation~\ref{obs:spxSatisfies}.  
For all $k \in [n]$, we have $v_i(S_k) \leq \frac{1}{(n+1-k)}v_i(M \setminus (S_{1} \cup S_{2} \cup \dots \cup S_{k-1}))$ by definition of $S_k$.  Applying this upper bound on $v_i(S_k)$ for $k = 1$, because $v_i(M) = 1$ we have that $v_i(M\setminus S_1) \geq 1 - \frac{1}{n} = \frac{n-1}{n}$.  By  applying the upper bound on $v_i(S_k)$ for $k = 2$ and our lower bound on $v_i(M \setminus S_1)$ we get $v_i(M \setminus (S_1 \cup S_2)) \geq \frac{n-1}{n} - \frac{1}{n-1} \cdot \frac{n-1}{n} \geq \frac{n-2}{n}$.  Iteratively repeating this process, we obtain that for all $k \in [n]$ we know that $v_i(M \setminus (S_{1} \cup S_{2} \cup \dots \cup S_{k})) \geq \frac{n-k}{n}$.  Also by definition,  we have that $v_i(S_\ell) + \text{min}_{j \in M \setminus (S_{1} \cup S_{2} \cup \dots \cup S_{\ell})}\{v_{ij}\} \geq \frac{1}{(n+1-\ell)} \cdot v_i(M \setminus (S_{1} \cup S_{2} \cup \dots \cup S_{\ell-1})) \geq \frac{1}{n+1-\ell} \cdot \frac{n-(\ell-1)}{n} = \frac{1}{n}$.  But finally, as long as the items from $S_{1} \cup S_{2} \cup \dots \cup S_{\ell-1}$ are not included in any of the bundles containing the items in $M \setminus (S_{1} \cup S_{2} \cup \dots \cup S_{\ell})$ in the complete allocation $X$, we have that $d_i(X) \geq \text{min}_{j \in M \setminus (S_{1} \cup S_{2} \cup \dots \cup S_{\ell})}\{v_{ij}\}$ so $i$ is $\propXt$-satisfied when allocated set $S_{\ell}$.
\end{proof}

%Using these observations and constructions, we recursively construct a $\propXt$ solution for $n$ agents.  Our construction selects an arbitrary agent to be the ``divider'' who finds $n$ sets $\{S_i\}_{i \in [n]}$ as in Theorem \ref{thm:recSpx}.  We then aim to find an assignment of the remaining agents to these  such that we terminate with a $\propXt$ solution or a partition of these sets and the remaining agents such that we can recursively find a $\propXt$ solution on the subproblems. \pnote{Theorem 4 provides the crux for how we are always able to do this. It gives us the flexibility to always insert the divider agent into partitions where not enough agents have high value for the partition with the guarantee that the divider will be satisfied by PROPm}  By leveraging Observation \ref{obs:disjointAlloc}, these recursive solutions remain a $\propXt$ solution in the initial problem of $n$ agents. 

\begin{lemma}
The divider agent is always PROPm-satisfied.  All non-divider agents are always PROPm-satisfied as well.
\end{lemma}
\begin{proof}
Note that the divider agent always receives a bundle $S_t$ in some iteration $t$. All the items from bundles $S_1, \dots, S_{t-1}$ are allocated to the agents that were in the decomposition $D$ at that time, while all the items from bundles $S_{t+1},\dots,S_n$ are allocated to the agents that were in $N_R$ at the time (and hence not in $D$). Then, given Theorem~\ref{thm:recSpx}, we conclude that the divider agent is always PROPm-satisfied.

Now observe that no agent other than the divider agent is directly allocated a bundle by our algorithm. Instead, the allocation to the other agents is decided recursively in some recursive call of a smaller sub-problem, when they are assigned the role of the divider. The important thing to verify is that PROPm-satisfying these agents in a recursive call, based on a subset of the agents and a subset of the goods, does, in fact, imply that they are PROPm-satisfied with respect to the original problem instances as well. 

In order to ensure this fact, we combine the statement of Observation~\ref{obs:disjointAlloc} with the definition of proportional sub-problems and decompositions. In particular, our definition of proportionality for a sub-problem guarantees that the conditions of Observation~\ref{obs:disjointAlloc} are met. Since we ensure that proportionality is maintained after every execution of the \textsc{UpdateDecomposition} subroutine, we guarantee that the combination of PROPm allocations for the generated sub-problems yields a PROPm allocation for the original problem.
\end{proof}

%% file: polynomialTime.tex
We now move to demonstrate that Algorithm \ref{alg:PROPm} completes in time polynomial in the number of agents and items. Lines 1 through 9 correspond to the ``divide phase'' and lines 10 through 14 correspond to the ``conquer phase''. % Observe that the sum of the number of agents in the recursive calls in the conquer phases is $n-1$ and the number of items is at most $m-1$ so it just remains to show that the dividing phase takes polynomial time.

The running time of Algorithm \ref{alg:PROPm} can be expressed as 
\begin{equation*}
    T(m,n) = f(m,n) + \sum_{j=1}^k T(m_j,n_j),
\end{equation*}
where $f(m,n)$ denotes the cost of the main call with $m$ items and $n$ agents, and the sum captures the cost of the recursive calls. The number of recursive calls is $k$ (equal to the number of sub-problems from lines \ref{algline:leftsubproblems} and \ref{algline:rightsubproblems}), while $m_j$ and $n_j$ are the number of items and agents, respectively, of the $j$-th sub-problem. Since all the sub-problems consist of distinct bundles and agents, we must have $\sum_{j=1}^k m_j \leq m-1$ and $\sum_{j=1}^k n_j \leq n-1$. Thus the width of any level of the recursion tree is at most $\max\{m,n\}$. Furthermore, since the size of each sub-problem strictly decreases through the recursion, the recursion tree has at most depth $\min\{m,n\}$. This means the total number of vertices in the recursion tree is polynomial in $n$ and $m$. All that remains is to show $f(m,n)$ is polynomially bounded in $m$ and $n$.

%First observe that producing the simple bundles of the divider takes $O(m \log m)$ time.  This is because sorting the items in non-decreasing order of value takes $O(m \log m)$ time and once the items are sorted a linear pass over the items suffices to generate the simple bundles.  It remains to verify that the body of the for loop takes polynomial time.

Producing the divider's bundles takes polynomial time since it requires only sorting the items in non-decreasing order of value and a linear pass over the sorted items.  In the body of the for loop of Algorithm \ref{alg:PROPm}, computing the initial value of $c$ takes time linear in the number of agents and items by asking each agent their value for each item in $S_1 \cup \dots \cup S_t$.  This initial value of $c$ is at most $n-1$.  As demonstrated in subsection~\ref{sec:correctSubroutine}, at each iteration of the while loop (beginning at line 6) in Algorithm \ref{alg:PROPm}, the \textsc{UpdateDecomposition} subroutine either increases the value of $|D.agents|$ from $t-1$ to $t$ or decreases the value of $c$.  Thus, the number of iterations of the while loop is at most $n$ at any iteration of the for loop.  

We now demonstrate that the body of the while loop takes polynomial time. Note that computing the new value of $c$ after then \textsc{UpdateDecomposition} subroutine takes linear time (in the number of agents and items). To verify that \textsc{UpdateDecomposition} also takes polynomial time, observe that the sub-problem graph induced by $D$ in iteration $t$ of the for loop contains at most $t-1$ vertices (which would occur when all $t-1$ agents in $D$ and bundles are assigned to distinct sub-problems).  We then add two additional vertices $w_\alpha$ and $w_\beta$ so there are at most $t + 1 = O(n)$ vertices overall at any iteration of the for loop.  Note that checking if an edge exists between two vertices in the graph takes time linear in the number of agents and goods in the two sub-problems and each sub-problem has at most $n$ agents and $m$ items.  Thus, we can construct the graph in time $O(n\cdot(n + m))$.  Finally, computing the set $R$ of reachable vertices from $w_\alpha$ can be accomplished by a simple breadth-first-search which is known to take time linear in the number of the vertices and edges in the graph.  Since updating the decomposition just requires propagating changes along a path of length at most $n$, the entire \textsc{UpdateDecomposition} process takes polynomial time.

%% file: conclusion.tex
In this paper, we solve the problem of computing PROPm allocations among agents with additive valuations, but we leave open another interesting problem proposed in \cite{PROPm}: the question of the existence and computation of an \emph{average-EFx} (a-EFx) allocation. To determine if an agent is a-EFx satisfied by some allocation, we remove $i$'s least favorite item from each other agent's allocated bundle, and then ask that $i$'s value for her own bundle is at least as high as her average value for all the other agents' bundles. This notion is stronger than PROPm and may provide an interesting stepping stone toward the, much harder, EFx problem.